\documentclass[a4paper,USenglish,cleveref, autoref]{lipics-v2019}

\usepackage{nicefrac}
\newenvironment{proofidea}{\noindent {\bf Proof idea.}}{\qed\medskip}

\newcommand{\id}{\mbox{\rm id}}
\newcommand{\acc}{\mbox{\rm accept}}
\newcommand{\rej}{\mbox{\rm reject}}

\newcommand{\bbF}{\mathop{\mathbb{F}}}
\newcommand{\fq}{{\textstyle \bbF_q}}

\newcommand{\set}[1]{\left\{ #1 \right\}}

\newcommand{\seq}{\textsc{SumZero}}

\newcommand{\prover}{\textsf{p}}
\newcommand{\verifier}{\textsf{v}}
\newcommand{\cI}{\mathcal{I}}

\newcommand{\cL}{{\mathcal{L}}}
\renewcommand{\P}{{\mathcal{P}}}
\newcommand{\BAD}{{\cal{B}}}
\newcommand{\LD}{\mbox{\sf LD}}
\newcommand{\LH}{\mbox{\sf LH}}

\newcommand{\BPLD}{\mbox{\sf BPLD}}
\newcommand{\BPP}{\mbox{\sf BPP}}
\newcommand{\PLS}{\mbox{\sf PLS}}
\newcommand{\RPLS}{\mbox{\sf RPLS}}
\newcommand{\LCP}{\mbox{\sf LCP}}
\newcommand{\NLD}{\mbox{\sf NLD}}

\newcommand{\MA}{\mbox{\sf dMA}}

\newcommand{\AM}{\mbox{\sf dAM}}
\newcommand{\AMH}{\mbox{\sf dAMH}}
\newcommand{\MAM}{\mbox{\sf dMAM}}
\newcommand{\AMA}{\mbox{\sf dAMA}}
\newcommand{\AMAM}{\mbox{\sf dAMAM}}

\newcommand{\niso}{\overline{\mathsf{Iso}}}
\newcommand{\sym}{\mathsf{Sym}}
\newcommand{\nsym}{\overline{\mathsf{Sym}}}
\newcommand{\ncol}{\overline{\mathsf{3Col}}}
\newcommand{\equ}{\mathsf{Eq}}
\newcommand{\nequ}{\overline{\mathsf{Eq}}}
\newcommand{\optim}{\mathsf{OptVal}}

\newcommand{\adm}{\mathsf{Adm}}
\newcommand{\tfree}{\Delta_\mathsf{free}}

\usepackage{hyperref}
\usepackage{tikz}
\usetikzlibrary{shapes}
\usepackage{eurosym}

\title{Trade-offs in Distributed Interactive Proofs}

\author{Pierluigi Crescenzi}{IRIF---CNRS and Universit\'e de Paris, France}%
{piluc@irif.fr}{}{}

\author{Pierre Fraigniaud}{IRIF---CNRS and Universit\'e de Paris, France}%
{pierref@irif.fr}{}{Additional support from the INRIA project GANG, and from the ANR project DESCARTES.}

\author{Ami Paz}{IRIF---CNRS and Universit\'e de Paris, France}%
{amipaz@irif.fr}{}{Supported by the Fondation des Sciences Math\'ematiques de Paris.}

\authorrunning{P. Crescenzi, P. Fraigniaud and A. Paz}
\Copyright{Pierluigi Crescenzi, Pierre Fraigniaud and Ami Paz}

%TODO mandatory; please add comma-separated list of keywords

\relatedversion{}%optional, e.g. full version hosted on arXiv, HAL, or other respository/website
\acknowledgements{The authors are thankful to Gianlorenzo D'Angelo for fruitful discussions on the topic of this paper, and to Amir Yehudayoff for discussion on his work~\cite{RY15}.}

\nolinenumbers %uncomment to disable line numbering

\hideLIPIcs  %uncomment to remove references to LIPIcs series (logo, DOI, ...), e.g. when preparing a pre-final version to be uploaded to arXiv or another public repository

\begin{document}
%%%%%%%%%%%%%%%%%%%%%%%%%%%%%%%%%%%%%%%%%%%

\maketitle

\begin{abstract}
The study of interactive proofs in the context of distributed network computing is a novel topic, recently introduced by Kol, Oshman, and Saxena [PODC 2018]. In the spirit of sequential interactive proofs theory, we study the power of distributed interactive proofs. This is achieved via a series of results establishing trade-offs between various parameters impacting the power of interactive proofs, including the number of interactions, the certificate size, the communication complexity, and the form of randomness used. Our results also connect distributed interactive proofs with the established field of distributed verification. In general, our results contribute to providing structure to the landscape of distributed interactive proofs.
\end{abstract}

%%%%%%%%%%%%%%%%%%%%%%%%%%%%%%%%%%%%%%%%%%%
\section{Introduction}
%%%%%%%%%%%%%%%%%%%%%%%%%%%%%%%%%%%%%%%%%%%

This paper is concerned with distributed network computing, in which $n$ processing nodes occupy the $n$ vertices of a connected simple graph $G$, and communicate through the edges of $G$. In this context, \emph{distributed decision}~\cite{NS95} refers to the task in which the nodes have to collectively decide whether the network $G$ satisfies some given graph property, which may refer also to input labels given to the nodes (basic examples of such tasks are whether the network is acyclic or whether the network is properly colored). If the property is satisfied then all nodes must accept, otherwise at least one node must reject. Distributed decision finds immediate applications to distributed fault-tolerant computing, in which the nodes must check whether the current network configuration is in a legal state with respect to some Boolean predicate~\cite{KKP10}. (If this is not the case, the rejecting node(s) may raise an alarm or launch a recovery procedure.) 

While some properties (e.g., whether a given coloring is proper) are locally decidable (LD) by exchanging information between neighbors only, other properties are not (e.g., whether the network is acyclic, or whether a given set of pointers forms a spanning tree of the network). As a remedy, the notion of \emph{proof-labeling scheme} (PLS) was introduced~\cite{KKP10}, and variants were considered, including  \emph{non-deterministic local decision} (NLD)~\cite{FKP13}, and \emph{locally checkable proofs} (LCP)~\cite{GS16}. All these settings assume the existence of a \emph{prover} assigning \emph{certificates} to the nodes, and a distributed \emph{verifier} in charge of verifying that these certificates form a distributed proof that the network satisfies some given property. For instance, acyclicness can be certified by a prover picking one arbitrary node $u$, and assigning to each node $v$ a certificate $c(v)$ equal its distance to $u$. The distributed verifier running at every node $v$ checks that $v$ has one neighbor $w$ satisfying $c(w)=c(v)-1$, and all its other neighbors $w'$ satisfying $c(w')=c(v)+1$. If the network contains a cycle, then these equalities will be violated in at least one node. 

Note that the prover is not necessarily an abstract entity, as an algorithm constructing some distributed data structure (e.g., a spanning tree) may construct in parallel a proof that this data structure is correct. Interestingly, all (Turing decidable) graph properties can be certified by PLS and LCP with $O(n^2)$-bits certificates~\cite{KKP10}, and this is tight~\cite{GS16} --- for instance, symmetry\footnote{$G$ is symmetric if $G$ has a non trivial automorphism, i.e., a one-to-one mapping from the set of nodes to itself preserving edges, and distinct from the identity map. } ($\sym$) was shown to require $\Omega(n^2)$-bit certificates. However, not only such universal certification requires high space complexity at the nodes for storing large certificates, it also requires high communication complexity between neighbors for verifying the correctness of these certificates. Hence, the concern is minimizing the size of the certificates for specific graph properties, e.g., minimum-weight spanning trees (MST)~\cite{KK07}. 

Recently, the notion of \emph{randomized} proof-labeling schemes (RPLS) was introduced~\cite{BFPS15}. RPLS assumes that the distributed verifier is randomized, and the global verdict provided by the nodes about the correctness of the network configuration should hold with probability at least~$\nicefrac23$. Such randomized distributed verification schemes were proven to be very efficient in terms of communication complexity (with $O(\log n)$-bit messages exchanged between neighbors), but this often holds at the cost of actually \emph{increasing} the size of the certificates provided by the prover. 

Another recent direction for reducing the certificate size introduces a \emph{local hierarchy} of complexity classes defined by \emph{alternating quantifiers} (similarly to the polynomial hierarchy~\cite{S77}) for local decision~\cite{FFH16}. Interestingly, many properties requiring $\Omega(n^2)$-bit certificates with a locally checkable proof stand at the bottom levels of this hierarchy, with $O(\log n)$-bit certificates. This is for instance the case of non 3-colorability ($\ncol$), which stands at the second level of the hierarchy, and $\sym$, which stands at the third level of the hierarchy. More generally, all monadic second order graph properties belong to this hierarchy with $O(\log n)$-bit certificates. However, it is not clear how to implement the protocols resulting from this hierarchy. 

Even more recently, a very original and innovative approach was adopted~\cite{KOS18,NPY18}, bearing similarities with the local hierarchy but perhaps offering more algorithmic flavor. 
This approach considers \emph{distributed interactive proofs}. Such proofs consist of a constant number of interactions between a centralized prover $\mathsf{M}$ (a.k.a.~Merlin) and a randomized distributed verifier $\mathsf{A}$ (a.k.a.~Arthur). For instance, a $\AM$ protocol is  a protocol with two interactions: 
Arthur queries Merlin by sending a random string, and Merlin replies to this query by sending a certificate. Similarly, a $\MAM$ protocol involves three interactions: Merlin provides a certificate to Arthur, then Arthur queries Merlin by sending a random string, and finally Merlin replies to Arthur's query by sending another certificate. This series of interactions is followed by a phase of distributed verification performed between every node and its neighbors, which may be either deterministic or randomized.

Although the interactive model seems weaker than the alternation of quantifiers in the local hierarchy, many properties requiring $\Omega(n^2)$-bits certificates with a locally checkable proof admit an Arthur-Merlin protocol with small certificates, and very few interactions. For instance, this is the case of $\sym$ that admits a $\MAM$ protocol with $O(\log n)$-bit certificates, and a $\AM$ protocol with $O(n\log n)$-bit certificates~\cite{KOS18} (on the other hand, any $\AM$ protocol for $\sym$ requires $\Omega(\log\log n)$-bit certificates~\cite{KOS18}). It is also known that non symmetry ($\nsym$) can be decided by a $\AMAM$ protocol with $O(\log n)$-bit certificates~\cite{NPY18}. These results raise several interesting questions, such~as: 

\begin{enumerate}
    \item Are there ways to establish trade-offs between space complexity (i.e., the size of the certificates) and communication complexity (i.e., the size of the messages exchanged between nodes)? The $\MA$ protocols~\cite{KOS18} as well as the RPLS protocols~\cite{BFPS15} enable to gain a lot in terms of message complexity, but at the cost of still high space complexity. Would it be possible to compromise between these two complexities? In particular, would it be possible to reduce the certificate size at the cost of increasing the communication complexity? 
    
    \item The theory of distributed decision has somehow restricted itself to distributed randomness, in the sense that each node has only access to a private source of random coins. These coins are public to the prover, but remain private to the other nodes. Shared randomness is known to be stronger than private randomness for communication complexity, as witnessed by, e.g., deciding equality~\cite{KN97}. How much shared randomness could help in the context of distributed decision?  

    \item Last, but not least, are there general reduction theorems between Arthur-Merlin classes for trading the number of interactions with the certificate size? Indeed, it is known that, in the centralized setting, $\mathsf{AM}[k]=\mathsf{AM}[2]$ for any $k\geq 2$, but it is not known whether such collapse holds in the distributed setting~\cite{KOS18}. Also, the Sipser–Lautemann theorem tells us that, in the centralized setting, $\mathsf{MA}\subseteq \Sigma_2\,\cap \, \Pi_2$ and $\mathsf{MA}\subseteq \mathsf{AM}\subseteq \Pi_2$. It is not known whether the distributed Arthur-Merlin classes stand so low in the local alternating hierarchy too. 
\end{enumerate}

%-------------------------------------
\paragraph*{Our Results.}
%-------------------------------------

In this paper, we study the power of distributed interactive proofs. This is achieved via a series of results establishing trade-offs between various parameters impacting the power of interactive proofs, including the number of interactions, the certificate size, the communication complexity, and the form of randomness used. Our results also connect distributed interactive proofs with the established field of distributed verification. We address the above three questions as follows. 
For the first question, 
we show how to apply techniques developed in the framework of multi-party communication complexity to get trade-offs between space and communication for the classical triangle detection problem. 
For the second question, 
we show that shared randomness helps significantly, enabling to exponentially reduce the communication complexity while preserving the space complexity for important problems such as spanning tree, and a vast class of optimization problems, including, for example, maximum independent set and minimum dominating set. 
For the third question, 
we give a general technique for reducing the number of interactions at the cost of increasing the certificate and message size. 

More specifically, for the first question, we explore the trade-off of space vs.\ communication, and establish that, for every $\alpha$, there exists a Merlin-Arthur protocol for triangle-freeness, using $O(\log n)$ bits of shared randomness, with $\widetilde{O}(n/\alpha)$-bit certificates and $\widetilde{O}(\alpha)$-bit messages between nodes (see Theorem~\ref{theo:triaglefreeness}). To our knowledge, this is the first example of a decision task for which one can trade communication for space, and vice versa. In addition, the proof reveals an interesting connection between $\MA$ and communication complexity with a referee. Note that, for $\alpha=\sqrt{n}$, we obtain a distributed Merlin-Arthur protocol for triangle-freeness with message and space complexities $\widetilde{O}(\sqrt{n})$ bits. In contrast, any proof-labeling scheme for triangle-freeness must have certificate size at least $n/e^{O(\sqrt{\log n})}$ bits (see Proposition~\ref{prop: lb for triangle-freeness}). A similar tradeoff can be obtained when using distributed randomness, though with higher space complexity.

% The proof of Theorem~\ref{result1} makes use of sequential Merlin-Arthur protocols for deciding disjointness in the setting of 2-party communication complexity with a referee~\cite{ARR17,AW09}. 

Regarding the second question, we explore the significance of having access to shared randomness. We show that, for any minimization problem $\pi$ in graphs whose admissibility can be decided locally, there exists a Merlin-Arthur protocol for certifying the existence of a solution whose cost is at most~$k$, using $O(\log n)$ bits of shared randomness, with $O(\log n)$-bit certificates and $O(\log\log n)$-bit messages between nodes (see Theorem~\ref{theo:optimization}). The same result holds for maximization problems whose admissibility can be decided locally. Note that this class of problems includes, for example, maximum independent set, minimum dominating set, and minimum vertex cover (potentially weighted). This exponentially improves the communication complexity of locally checkable proofs for such problems. The same communication complexity could be obtained using randomized proof-labeling schemes, but at the cost of increasing the certificate size to up to $O(n\log n)$ bits. 
%
% The proof of Theorem~\ref{result2} is obtained by using the multi-party communication complexity protocol with a referee for solving the $\seq$ problem by~\cite{N93}. 
%
As another interesting result in the context of exploring the significance of having access to shared randomness, we show that even shared randomness remains limited both in terms of the certificate size and of the amount of communication. We show that every Arthur-Merlin protocol for $\sym$, and every Arthur-Merlin protocol for $\nsym$ must have both certificate size and message size $\Omega(\log\log n)$ bits, even with shared randomness (see Theorem~\ref{thm:lowerboundsymmetry}).
%
% This result is obtained using a simplified version of the proof of the same result for $\sym$ in~\cite{KOS18}, and generalizing it to $\nsym$, by using a lower bound on the communication complexity of Arthur-Merlin games for the equality and non-equality problems in~\cite{GPW16}.
%
Interestingly, for the class of graphs used in the proof of this latter result, there is a Merlin-Arthur protocol with certificates and messages of constant length. This shows that the inclusion $\mathsf{MA}\subseteq \mathsf{AM}$ which holds in the centralized setting does not hold in the distributed setting. 

Finally, we consider general reductions within the Arthur-Merlin hierarchy, and compare the power of this hierarchy to the power of proof-labeling schemes with certificates of linear size. We show that, for every  $\sigma$ and $\gamma$, any graph property verifiable with an Arthur-Merlin protocol with 3 or 4~interactions ($\MAM$ or $\AMAM$) using $\sigma$-bit certificates and  $\gamma$-bit messages can also be verified by an Arthur-Merlin ($\AM$) protocol using $O(n\sigma^2)$-bit certificates and  $O(n\gamma\sigma)$-bit messages (see Theorem~\ref{thm:reduction} and Corollary~\ref{cor:reduction}). Although the linear blowup in terms of both certificate size and message complexity may seem huge at a first glance, it fits (up to logarithmic factors) with the different results obtained previously~\cite{KOS18,NPY18} regarding $\sym$ and graph non-isomorphism ($\niso$). Finally, we compare the power of Arthur-Merlin protocols with an arbitrarily large number of interactions with the power of proof-labeling schemes. We show that there exists a graph property admitting a proof-labeling scheme with certificates and messages on $O(n)$ bits, that cannot be solved by an Arthur-Merlin protocol with $o(n)$-bit certificates, for any fixed number $k\geq 0$ of interactions between Arthur and Merlin, even using shared randomness, and even with messages of unbounded size (see Theorem~\ref{thm:insideSigma1}). This latter result demonstrates that, in general, one cannot trade the number of interactions between Merlin and Arthur for reducing the certificate size, at least for certificates of linear size. 
%
% Its proof uses a similar counting argument as the one used in~\cite{FFH16} for showing that the local hierarchy does not contain all distributed languages. 
%

Most of our results are stated by assuming that nodes have access to shared randomness. However, as all our protocols are local, all our 1-round protocols can be simulated by 2-round protocols using distributed randomness. In general, our results contribute to providing structure to the landscape of distributed interactive proofs. 

%-------------------------------------
\paragraph*{Related Work.}
%-------------------------------------

Local decision (\LD) and the central notion of locally-checkable labellings were introduced and thoroughly studied in the 90s~\cite{NS95}. Local verification was introduced fifteen years later~\cite{KKP10}, through the original notion of proof-labeling schemes (\PLS).  Proof-labeling schemes find important applications to self-stabilization, but are subject to some restrictions (only certificates are exchanged between neighbors). These restrictions were lifted  by considering the general notion of locally checkable proofs (\LCP)~\cite{GS16}. By definition, we have $\LCP=\Sigma_1\LD$ (the same way $\mathsf{NP}=\Sigma_1\mathsf{P}$). A third notion of distributed verification was introduced, by considering the class \NLD~\cite{FKP13}. \NLD\/ differs from $\Sigma_1\LD$ in the fact that, in \NLD, the certificates cannot depend on the identities given to the nodes. 
 
Randomized versions of local decision and local verification have been considered in the literature \cite{BFPS15,FGKPP14,FKP13,KOS18}. A Merlin-Arthur (\MA) protocol is actually a randomized version of locally checkable proof ($\Sigma_1\LD$) that was previously studied~\cite{BFPS15}: Merlin provides each node with a certificate, and Arthur performs a randomized verification algorithm at each node. The benefit  of using $\MA$ over $\Sigma_1\LD$ can be exponential in terms of communication complexity (i.e., of the size of the messages exchanged between neighbors), at the cost of a linear increase in space complexity (i.e., of the size of the certificates provided by Merlin)~\cite{BFPS15}. 

Very closely related to the line of work about interactive distributed proofs is the local hierarchy 
$
\LH=\bigcup_{k\geq 0}\Big(\Sigma_k\LD \cup \Pi_k\LD\Big)
$
with certificates and messages of logarithmic size~\cite{FFH16}, extending the known $\mathsf{logLCP}$ class~\cite{GS16}. In particular, it is proved that $\sym\in\Sigma_3\LD$~\cite{FFH16}. Also, it is easy to show that $\ncol\in\Pi_2\LD$. In contrast, placing $\sym$ or $\ncol$ in $\Sigma_1\LD$ requires $\Omega(n^2)$-bit certificates~\cite{GS16}. The same hierarchy was later considered,  but under the constraint that the certificates must not depend on the identifier assignment to the nodes. With $O(n^2)$-bit certificates, this hierarchy collapses at the second level~$\Pi_2\LD$~\cite{BDFO18}. (This is in contrast with the hierarchy in which certificates can depend on the node identifiers, which collapses at the first level~$\Sigma_1\LD$ with $O(n^2)$-bit certificates.) Nevertheless, apart from the bottom levels, the hierarchies based on alternating quantifiers with $O(\log n)$-bit certificates are essentially the same~\cite{BDFO18}. See the recent survey~\cite{FF16} for more 
results on distributed decision.

This work is inspired by very recent achievements in the field of distributed interactive protocols~\cite{KOS18,NPY18}. In addition to the aforementioned results regarding $\sym$ and $\nsym$, two  versions of $\niso$ (Given two sub-graphs $G_1$ and $G_2$ of $G$, $G_1\not\sim G_2$, that is, are $G_1$ and $G_2$ non-isomorphic?) were considered. For the easiest version, in which the input of each node $v$ is formed by the two sets of neighbors of $v$ in the two graphs $G_1$ and $G_2$, a $\AMAM$ protocol with $O(\log n)$-bit certificates for deciding $G_1\not\sim G_2$ was designed~\cite{NPY18}. For the more complicated version, in which $G_1=G$ (that is, $G_1$ is the communication network) and the input of each node $v$ is the set of neighbors of $v$ in $G_2$,  a $\AMAM$ protocol with $O(n \log n)$-bit certificates for deciding $G_1\not\sim G_2$ was proposed~\cite{KOS18}, and  an Arthur-Merlin protocol with a constant number of interactions and $O(\log n)$-bit certificates, for deciding $G_1\not\sim G_2$, was successively designed~\cite{NPY18}.
% For the easiest version, $\mbox{input}(v)=(E_{G_1}(v),E_{G_2}(v))$ where $E_{G_i}(v)\subseteq E_G(v)$ for every $i\in\{1,2\}$, a $\AMAM$ protocol with $O(\log n)$-bit certificates for deciding $G_1\not\sim G_2$ was designed~\cite{NPY18}. For the more complicated version, $G_1=G$ and $\mbox{input}(v)=E_{G_2}(v')$,  a $\AMAM$ protocol with $O(n \log n)$-bit certificates for deciding $G_1\not\sim G_2$ was proposed~\cite{KOS18}, and  a $\AM[k]$ protocol with $k=O(1)$ interactions, and $O(\log n)$-bit certificates, for deciding $G_1\not\sim G_2$ was designed~\cite{NPY18}. 
Interestingly, this latter result is obtained via a general connection between efficient centralized computation (under various models) and the ability to design Arthur-Merlin protocols with a constant number of interactions between the prover and the verifier, using logarithmic-size certificates. 

We use a variety of techniques and results from the theory of \emph{communication complexity}~\cite{KN97}. Specifically, we use an Arthur-Merlin style protocol for two-party disjointness~\cite{AW09}, in a recent variant~\cite{ARR17} that allows a trade-off between communication complexity and certificate size.
We also use recent lower bounds for the equality and non-equality problems in the same setting~\cite{GPW16}.
Finally, we use multi-party communication protocol with a referee for the $\seq$ problem with bounded inputs~\cite{N93}, and with unbounded inputs~\cite{KLM18}.

%%%%%%%%%%%%%%%%%%%%%%%%%%%%%%%%%%%%%%%%%%%
\section{Model and Definitions}
%%%%%%%%%%%%%%%%%%%%%%%%%%%%%%%%%%%%%%%%%%%

A \emph{network configuration} is a triple $(G,\id,x)$ where $G=(V,E)$ is a connected simple graph, $\id:V\to \{1,\dots,n^c\}$ for some constant $c\geq 1$ is the \emph{identity} one-to-one assignment to the nodes, and $x:V\to\{0,1\}^*$ is the \emph{input label} assignment (i.e., the state of the node). A \emph{distributed language} is a collection $\cL$ of network configurations. Note that it may be the case that, for some language $\cL$, $(G,\id,x)\in\cL$ while $(G,\id',x)\notin\cL$ for two different identity assignments $\id$ and $\id'$. This typically occurs for languages where the label of a node refers to the identities of its neighbors, e.g., for encoding spanning trees. (Throughout the paper, we assume that all considered distributed languages are Turing-decidable). \emph{Distributed decision} for $\cL$ is the following task: given any network configuration $(G,\id,x)$, the nodes of $G$ must collectively decide whether $(G,\id,x)\in \cL$. If this is the case, then \emph{all} nodes must accept, otherwise \emph{at least one node} must reject (with certain probabilities, depending on the model). 

%-------------------------------------
%\subsection{Interactive Protocols}
%-------------------------------------

We consider \emph{interactive protocols} for distributed decision~\cite{KOS18}. A distributed interactive protocol $\P$ consists of a constant series of interactions between a \emph{prover} called \emph{Merlin}, and a \emph{verifier} called \emph{Arthur} (see Fig.~\ref{fig:dAM} for a visual representation of such a protocol). The prover Merlin is centralized, and has unlimited computing power. It is aware of the whole network configuration $(G,\id,x)$ under consideration, but it cannot be trusted. The verifier Arthur is distributed, and has bounded knowledge, that is, at each node $v$, Arthur is initially aware solely of $(\id(v),x(v))$, i.e., of its identity and its input label. 

\begin{figure}
	\begin{tikzpicture}[scale=0.475,every node/.style={scale=0.475,minimum width=1.1cm}]
	\draw (-7,0) rectangle  +(14,4.5);
	\path (-7,0) --  (7,0) node[midway,above] {\Large{shared randomness}};
	\node[cloud, draw,cloud puffs=20,cloud puff arc=60, aspect=6, inner ysep=2em] at (0,3) {};
	\node[circle,draw,double, double distance=1pt] (euro) at (-4,1) {\Large{\euro}};
	\node[rectangle,draw,double, double distance=1pt] (merlin) at (0,1) {\Large{Merlin}};
	\draw[->] (euro) -- node[below] {$r_j$} (merlin);
	\foreach \w/\x/\y/\z/\p/\d/\l in {1/-4/3/1/{0.3}/{left}/{-0.3cm},2/-1.5/3/2/{0.3}/{left}/{-0.2cm},3/1.5/3/{n-1}/{0.3}/{left}/{-0cm},4/4/3/n/{0.1}/{right}/{0.2cm}} {
		\node[circle,draw] (v\w) at (\x,\y) {\Large{$v_{\z}$}};
		\draw[->] (euro) -- node[\d=\l,pos=\p] {\Large{$r_j$}} (v\w);
	}
	\node at (0,3) {$G$};
	\node at (4,1) {$|r_j|=\rho(n)$};
	\begin{scope}[yshift=-7cm]
		\draw (-7,0) rectangle  +(14,6.5);
		\path (-7,0) --  (7,0) node[midway,above] {\Large{distributed randomness}};
		\node[cloud, draw,cloud puffs=20,cloud puff arc=60, aspect=6, inner ysep=2em] at (0,3) {};
		\node[rectangle,draw,double, double distance=1pt] (merlin) at (0,1) {\Large{Merlin}};
		\foreach \w/\x/\y/\z/\p/\d/\l/\q/\e/\m in {1/-4/3/1/{0.3}/{left}/{-0.1cm}/{0.3}/{left}/{-0cm},2/-1.5/3/2/{0.3}/{left}/{-0.1cm}/{0.3}/{left}/{-0.3cm},3/1.5/3/{n-1}/{0.3}/{right}/{0cm}/{0.3}/{right}/{0cm},4/4/3/n/{0.3}/{right}/{-0cm}/{0.3}/{right}/{-0cm}} {
			\node[circle,draw] (v\w) at (\x,\y) {\Large{$v_{\z}$}};
			\node[circle,draw,double, double distance=1pt] (euro\w) at (\x,\y+2.5) {\Large{\euro}};
			\draw[->] (euro\w) -- node[\d=\l,pos=\p] {\Large{$r_j(v_{\z})$}} (v\w);
			\draw[->] (v\w) -- (merlin);
		}
		\node at (0,3) {\Large{$G$}};
		\node at (4,1) {\Large{$|r_j(v_i)|=\rho(n)$}};
	\end{scope}
	\draw (-7.5,-8) rectangle  +(15,13);
	\path (-7.5,-8) --  (7.5,-8) node[midway,above] {\Large{(a) Arthur phase with random complexity $\rho$ at interaction $j$}};
	\begin{scope}[xshift=15cm]
		\draw (-7,0) rectangle  +(14,4.5);
		\path (-7,0) --  (7,0) node[midway,above] {\Large{(b) Merlin phase with space complexity $\sigma$ at interaction $j$}};
		\node[cloud, draw,cloud puffs=20,cloud puff arc=60, aspect=6, inner ysep=2em] at (0,3) {};
		\node[rectangle,draw,double, double distance=1pt] (merlin) at (0,1) {\Large{Merlin}};
		\foreach \w/\x/\y/\z/\p/\d/\l in {1/-4/3/1/{0.7}/{left}/{0.2cm},2/-1.5/3/2/{0.7}/{left}/{-0cm},3/1.5/3/{n-1}/{0.7}/{left}/{-0cm},4/4/3/n/{0.7}/{right}/{0.1cm}} {
			\node[circle,draw] (v\w) at (\x,\y) {\Large{$v_{\z}$}};
			\draw[->] (merlin) -- node[\d=\l,pos=\p] {\Large{$c_j(v_{\z})$}} (v\w);
		}
		\node at (0,3) {\Large{$G$}};
		\node at (4,1) {\Large{$|c_j(v_i)|=\sigma(n)$}};
	\end{scope}
	\begin{scope}[xshift=15cm,yshift=-7cm]
		\draw (-7,0) rectangle  +(14,4.5);
		\path (-7,0) --  (7,0) node[midway,above] {\Large{(c) Verification phase with communication complexity $\gamma$}};
		\node[cloud, draw,cloud puffs=20,cloud puff arc=60, aspect=6, inner ysep=2em] at (0,3) {};
		\foreach \w/\x/\y/\z/\p in {1/-4/3/1/{left},2/-1.5/3/2/{left},3/1.5/3/{n-1}/{right},4/4/3/n/{right}} {
			\node[circle,draw] (v\w) at (\x,\y) {\Large{$v_{\z}$}};
		}
		\draw[->] (v2) -- node[below] {\Large{$m_{2\rightarrow 1}$}} (v1);
		\draw[->] (v2) to[bend left=60] node[below=0.1cm] {\Large{$m_{2\rightarrow {n-1}}$}} (v3);
		\node at (0,3) {\Large{$G$}};
		\node at (4,1) {\Large{$|m_{i\rightarrow j}|=\gamma(n)$}};
	\end{scope}
\end{tikzpicture}
\caption{The three different phases of a distributed Arthur-Merlin protocol (the Arthur phase can make use of either shared or distributed randomness)}
\label{fig:dAM}
\end{figure}
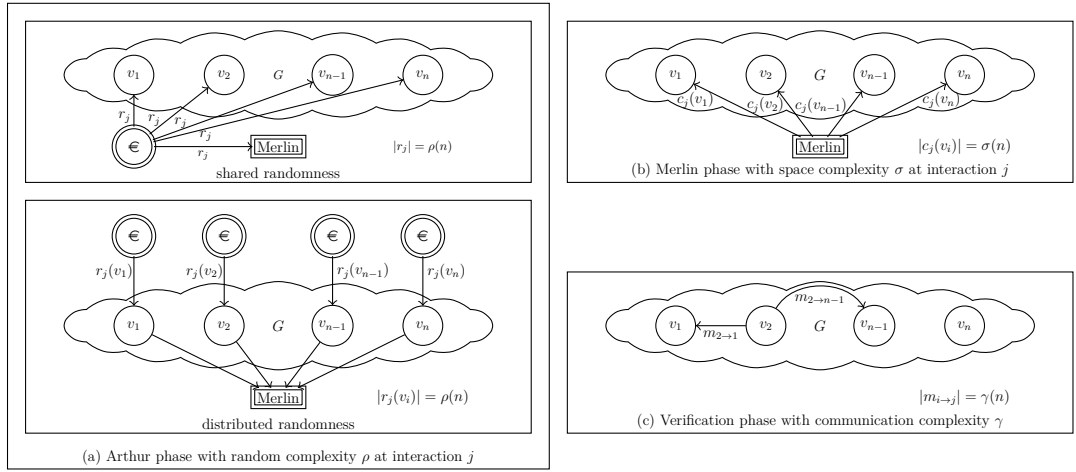

In a distributed \emph{Arthur-Merlin} interactive protocol performed on $\cI=(G,\id,x)$, whenever Arthur is the one that starts interacting, it picks a random string $r_1(v)$  at each node $v$ of $G$ (this random string might be private to each node, or the nodes may have access to shared randomness). Given the collection $r_1$ of random strings selected by the nodes, Merlin provides every node $v$ with a \emph{certificate} $c_1(v)=\prover(v,\cI,r_1)$, where $\prover:\{0,1\}^*\to\{0,1\}^*$. At this point, Arthur picks another random string $r_2(v)$  at each node~$v$. Then Merlin replies to each node~$v$ by sending a second certificate $c_2(v)=\prover(v,\cI,r_1,r_2)$, and so on. Whenever Merlin is the one that starts interacting, the process starts with Merlin constructing a binary string $c_0(v)=\prover(v,\cI)$ that it sends to every node $v$. These interactions proceed for a constant number $k\geq 0$ of times, and Merlin interacts last, by sending $c_{\lfloor\frac{k}{2}\rfloor}(v)$ to every node~$v$. A sequence of interactions can then be summarized by a \emph{transcript} $\pi(\cI,\prover,r)=\Big( \pi_v(\cI,\prover,r) \Big)_{v\in V}$, where $r=\left(r_1,\dots,r_{\lfloor\frac{k}{2}\rfloor}\right)$ with $r_i=(r_i(v))_{v\in V}$ for $i=1,\dots,\lfloor\frac{k}{2}\rfloor$, and 
\[
\pi_v(\cI,\prover,r)=\left(c_0(v),r_1(v),c_1(v),\dots,r_{\lfloor\frac{k}{2}\rfloor}(v),c_{\lfloor\frac{k}{2}\rfloor}(v)\right)
\]
with $c_0(v)=\varnothing$ if Arthur starts the interactions. In other words, an Arthur-Merlin protocol with $k$ interactions results in a transcript with $c_0(v)=\varnothing$ if $k$ is even, and with $c_0(v)$ equal to the first certificate provided by Merlin otherwise.  The Arthur-Merlin protocol completes by performing a \emph{deterministic} distributed verification algorithm $\verifier$ executed at each node. Specifically, Algorithm~$\verifier$ proceeds as follows at every node $v$: 
\begin{enumerate} 
\item A message $M_{v,u}$ destined for every neighboring node $u$ of $v$ is forged, and sent to $u$. This message may depend on the identity $\id(v)$, the input $x(v)$, all random strings generated by Arthur at $v$, and all certificates received by $v$ from Merlin. 
\item Based on all the knowledge accumulated by $v$ (i.e., its identity, its input label, the generated random strings, the certificates received from Merlin, and all the messages received from its neighbors), Algorithm~$\verifier$ accepts or rejects at node~$v$.
\end{enumerate}

A distributed Arthur-Merlin protocol $\P$ thus consists of two consecutive stages: (1)~interactions between the nodes and the prover~$\prover$ (Arthur-Merlin phases), and (2)~communication among neighboring nodes (algorithm~$\verifier$). Note that, for the sake of simplifying the presentation, and unifying the comparison with previous work, we restrict ourselves to verification algorithms $\verifier$ that perform in a single round. Performing more than one round enables to improve the complexity of verification protocols in some cases~\cite{FFHPP18}. However,  this does not conceptually change the nature of the protocol. For zero interactions (i.e., $k=0$), a distributed Arthur-Merlin protocol simply consists in performing a (deterministic) decision algorithm at each node~\cite{KKP10}. For one interaction (i.e., $k=1$), a distributed Arthur-Merlin verification protocol is a (1-round) locally-checkable proof algorithm~\cite{GS16}. 

\begin{definition}
The class $\AM[k](\sigma,\gamma)$ is the class of languages $\cL$ for which there exists a distributed Arthur-Merlin verification protocol with at most $k\geq 0$ interactions between Arthur and Merlin, where Merlin provides certificates of at most $\sigma\geq 0$ bits to the nodes, and the verification algorithm $\verifier$ exchanges messages of at most $\gamma\geq 0$ bits between nodes, such that, for every configuration $\cI=(G,\id,x)$, 
\[
\left\{\begin{array}{lcl}
(G,\id,x)\in \cL & \Rightarrow & \exists \prover : \Pr_r[\verifier(\pi(\cI,\prover,r))\; \mbox{accepts at all nodes}]\geq \nicefrac23;  \\
(G,\id,x)\notin \cL & \Rightarrow & \forall \prover : \Pr_r[\verifier(\pi(\cI,\prover,r))\;\mbox{rejects in at least one node}]\geq \nicefrac23.
\end{array}\right.
\]
\end{definition}

The definition of distributed \emph{Merlin-Arthur} interactive protocols, and of $\MA[k](\sigma,\gamma)$ is similar, apart from the fact that, as opposed to Arthur-Merlin protocols in which Merlin always interacts last, Arthur has one more ``interaction'' during which it picks a random bit-string $r'(v)$ at every node~$v$, which is used to perform a \emph{randomized} verification algorithm~$\verifier$. Therefore, for $k\geq 1$, a Merlin-Arthur protocol with $k$ interactions can be defined as an Arthur-Merlin protocol with $k-1$ interactions, but where the verification algorithm~$\verifier$ is randomized. For zero interactions, a distributed Merlin-Arthur protocol simply consists in performing a (deterministic) decision algorithm at each node~\cite{KKP10}. For one interaction, a distributed Merlin-Arthur protocol is a (1-round) randomized decision algorithm as studied previously~\cite{FGKPP14}. For two interactions, a distributed Merlin-Arthur protocol is a (1-round) randomized locally-checkable proof algorithm, also as studied previously~\cite{BFPS15}. 

In the following, we may avoid mentioning the parameters $\sigma$ and $\gamma$ when they are clear from the context, or when they are respectively identical in the two terms of an equality. For small values of $k\geq 2$, $\AM[k]$ and $\MA[k]$ are rewritten as an alternating sequence of $\mathsf{A}$s and $\mathsf{M}$s. For instance, $\AM[2]=\AM$, $\MA[2]=\MA$, $\AM[3]=\MAM$, $\MA[3]=\AMA$, and $\AM[4]=\AMAM$, and so on. For $k\leq 1$, it follows from the definition that 
$
\AM[0]=\MA[0]=\LD.
$
We also have 
$
\AM[1]=\Sigma_1\LD \; \mbox{and} \; \MA[1]=\BPLD(\nicefrac23,\nicefrac23)
$
where the class $\BPLD(p,q)$ is the distributed version of $\BPP$~\cite{G77}, with $p$ being the acceptance probability of the interactive protocol on legal instances, and $q$ being the rejection probability of the protocol on illegal instances~\cite{FGKPP14}. As a last example, $\MA$ is the class of languages that can be decided by a randomized locally checkable proof, as studied previously~\cite{BFPS15}. 

As opposed to the sequential setting in which it is known that $\mathsf{AM}[k]=\mathsf{AM}[2]$ for all $k\geq 2$, it is not known whether such ``collapse'' occurs in the distributed setting. Therefore, we define the Arthur-Merlin \emph{hierarchy} as $
\AMH(\sigma,\gamma)=\bigcup_{k\geq 0}\AM[k](\sigma,\gamma).
$
That is, $\cL\in\AMH(\sigma,\gamma)$ if and only if there exists $k\geq 0$ such that $\cL\in\AM[k](\sigma,\gamma)$. 

%-------------------------------------
%\subsection{On the Success Probability of Interactive Protocols}
% \medskip 
% \noindent\textit{Remark.} 
%-------------------------------------
\paragraph*{Boosting the Success Probability.}
In classical, sequential randomized algorithm, the success probability constant~$\nicefrac23$ can be easily increased using repetitions.
On the other hand, it was shown that this boosting technique is not  applicable for randomized distributed decision algorithms in general~\cite{FGKPP14}, making the choice of constant significant when considering such settings.
The inability of boosting in the distributed setting is due to the fact that, when repeating the algorithm on a ``no'' instance for several times, different nodes may reject in different repetitions, causing each node to sees very few rejections and decide on acceptance.
Somewhat surprisingly, we can show that in the case of distributed Arthur-Merlin protocols (i.e., $\AM[k]$ classes), parallel repetition is possible, and at a relatively low blowup in communication and certificates.
This allows us to boost the success probability, as follows.

\begin{proposition}\label{prop:boostingproba}
Let $1> p'>p>\nicefrac12$. If there exists an Arthur-Merlin verification protocol $\P$ with $k\geq 2$ interactions that enables to verify a distributed language $\cL$ with $\sigma$-bit certificates, $\gamma$-bit messages, and success probability~$p$, then  there exists an Arthur-Merlin verification protocol $\P'$ with $k$ interactions that enables to verify $\cL$ with $\sigma+O(\log n)$-bit certificates, messages on $\gamma+O(\log n)$ bits, and success probability $p'$.
\end{proposition} 

\begin{proof}
Moving from success probability $p$ to success probability $p'>p$ is achieved in a standard way, by merely repeating $\P$ a constant number of times (depending on $p$ and $p'$), and adopting the majority of the outcomes. However, this cannot be done in a straightforward manner because, for a configuration $(G,x,\id)\notin \cL$, it may be the case that the (at least one) node rejecting $(G,x,\id)$ is different at each repetition. Therefore, during the last interaction with the prover, Merlin provides every node with a local encoding of a spanning tree $T$ enabling to count the number of executions of $\P$ resulting in at least one node rejecting. It is known that certificates of $O(\log n)$ bits suffice for certifying such a tree~\cite{KKP10}. The root of the tree $T$ accepts or rejects depending on whether the majority of executions of $\P$ accepted or rejected, respectively.
\end{proof}

%We study this question in the case of distributed Merlin-Arthur protocols (i.e., $\MA[k]$ classes), the constant~$\nicefrac23$ set for the acceptance and rejection probabilities may not be arbitrary. For example, in the case of randomized distributed decision, it was shown that boosting the acceptance and rejection probabilities is not possible in general~\cite{FGKPP14}. However, in the case of 1-sided error (i.e., $(G,\id,x)\in \cL \; \Rightarrow \; \Pr[\mbox{all nodes accept}]=1$), one can simply boost the rejection probability to any arbitrarily large constant, by repetition. 
%
%On the other hand, in the case of $\AM[k]$ classes, we show that the constant $\nicefrac23$ set for the acceptance and rejection probabilities is arbitrary.
%, even for 2-sided error protocols, as shown below: 

%%%%%%%%%%%%%%%%%%%%%%%%%%%%%%%%%%%%%%%%%%%
\section{Space  vs. Communication }
%%%%%%%%%%%%%%%%%%%%%%%%%%%%%%%%%%%%%%%%%%%

In this section, we study the trade-off between space and communication complexity for Merlin-Arthur interactive protocols. Specifically, we consider the classical triangle-freeness problem, and establish a trade-off between space and communication for this problem. Recall that a graph $G=(V,E)$ is triangle-free if, for every three nodes $u,v,w$ in $V$, either $\{u,v\}\notin E$,  $\{u,w\}\notin E$, or $\{v,w\}\notin E$. We denote by $\tfree$ the corresponding distributed language. There is a recent deterministic distributed algorithm for triangle-freeness running in $\widetilde{O}(\sqrt n)$ rounds in the \textsc{congest} model~\cite{CPZ19}. A general scheme for designing $\MA$ protocol has also been proposed~\cite{BFPS15}. This scheme  enables to reduce communication complexity to $O(\log n)$, at the cost of increasing the space complexity to $O(n^2)$.
When more interactions are allowed, say a constant number $k$, a recent reduction~\cite{NPY18} --- from centralized small-space algorithms to our setting ---
implies a protocol with $O(\log n)$-bit certificates and $O(\log n)$-bit messages for the triangle-freeness problem, i.e., $\tfree \in \MA[k](\log n,\; \log n)$ for some constant $k>1$.
% %-------------------------------------
% \paragraph*{Merlin-Arthur protocol for triangle-freeness.}
% %-------------------------------------

In order to prove the trade-off between space and communication complexities for triangle-freeness, we first state the following result, which will be used at several places in the paper.

\begin{lemma}\label{lem:distance-2coloring}
For any network of maximum degree $d$, there exists a proof-labeling scheme (and thus a $\Sigma_1\LD$ protocol) with $O(\log n)$-bit certificates providing each node $v$ with the certified value $n$ of the number of nodes, and a color $c(v)\in\{1,\dots,\min\{d^2+1,n\}\}$ such that $c$ forms a certified proper distance-2 coloring of the network.
\end{lemma}

\begin{proofidea}
The certification of the number of nodes can be done by using a rooted spanning tree and by counting nodes in the sub-trees. The certification of a proper distance-2 coloring can be done by assigning colors to nodes, with every node checking that all its neighbors have different colors, all different from its own color. 
\end{proofidea}

Since triangle-freeness is a local property, nodes do not need to be represented by identifiers that are different throughout the entire network. Instead, identifiers resulting from a proper distance-2 coloring suffice. Therefore, using Lemma~\ref{lem:distance-2coloring}, we can assume that nodes are provided with identifiers in $\{1,\dots,n\}$ such that $\id(u)\neq\id(v)$ whenever the distance between $u$ and $v$ is at most $2$.

\begin{theorem}\label{theo:triaglefreeness}
For every $\alpha = O(n)$, there exists a Merlin-Arthur protocol for triangle-freeness, using $O(\log n)$ bits of shared randomness, with $O(\frac{n}{\alpha}\log n)$-bit certificates and $O(\alpha\log n)$-bit messages between nodes. In short $\tfree \in \MA(\frac{n}{\alpha}\log n,\; \alpha\log n)$.  
\end{theorem}

% \begin{proofidea}
% We first observe that verifying triangle-freeness is equivalent to verifying that the neighbor lists of any two adjacent nodes are disjoint. To decide disjointness we then utilize \emph{algebrization} techniques that are also utilized in the setting of communication complexity. (This exemplify an interesting connection between distributed verification and communication complexity with a referee, which we also utilize later.)
% \end{proofidea}

\begin{proof}
We identify the space $\{1,\dots,n\}$ of IDs with $[n/\alpha]\times[\alpha]$, for some $\alpha=O(n)$ of choice.
Each node $u$ thus has a set $S_u$ of pairs of the form $(i,t)$ representing its neighbors.
Let $q$ be a prime such that $cn\alpha <q\leq 2cn\alpha$,
for a large enough constant $c>1$, and let $\fq$ be the field of $q$ elements.
Each node $u$ represents $S_u$ as $\alpha$ functions
$
\psi_{S_u,t}:[n/\alpha]\to\set{0,1},
$
where
$
\psi_{S_u,t}(i)=1 \iff (i,t)\in S_u.
$
Node $u$ then extends these functions to polynomials
$
\Psi_{S_u,t}:\fq\to\fq
$
of degree at most $n/\alpha-1$ that agree with $\psi_{S_u,t}$ on $[n/\alpha]$. To make sure that an edge $\{u,v\}$ is not a part of a triangle,
the nodes $u$ and $v$ need to verify that $S_u\cap S_v=\emptyset$,
which is equivalent to $\Psi_{S_u,t}(i)\cdot\Psi_{S_v,t}(i)=0$
for all $i\in[n/\alpha]$ and $t\in[\alpha]$.
Node $u$ then defines its \emph{neighbors polynomials}
$
\Psi_{uv,t}=\Psi_{S_u,t}\cdot\Psi_{S_v,t}
$
for every $v\in S_u$, and every $t\in[\alpha]$. Let
$
\Psi_u=\sum_{t\in[\alpha]}\sum_{v\in S_u} \Psi_{uv,t}.
$
The degree of each polynomial $\Psi_{uv,t}$ is at most $2(n/\alpha-1)$,
and thus this is also the case for the degree of $\Psi_u$.
Node $u$ is not part of a triangle
if and only if $\Psi_{uv,t}(i)=0$ for every $t\in[\alpha]$, $i\in[n/\alpha]$ and $v\in S_u$. Since $q>n\alpha$, it follows that $u$ is not part of a triangle if and only if $\Psi_u(i)=0$ for every $i\in[n/\alpha]$.
(For each $i$, we have a sum of $n\alpha$ values, each in $\set{0,1}$.)

Merlin assigns to node $u$ the certificate $\Phi_u$, which is supposed to be equal to $\Psi_u$.
Since this is a polynomial of degree at most 
$2(\frac{n}{\alpha}-1)$, the same number of coefficients are sufficient for representing $\Phi_u$. Therefore, the certificates are of $O(\frac{n}{\alpha}\log q)$ bits,
which are actually $O(\frac{n}{\alpha}\log n)$ bits, as $q\leq 2cn\alpha=O(n^2)$.

Each node $u$ first verifies that $\Phi_u(i)=0$ for every $i\in[n/\alpha]$.
Then, it checks that indeed $\Phi_u=\Psi_u$, as follows.
The protocol uses the shared randomness to choose a field element $i_0\in\fq$ known to all nodes.
Each node $v$ broadcasts $\{\Psi_{S_v,t}(i_0):t\in[\alpha]\}$ to each of its neighbors, using $O(\alpha \log q)\leq O(\alpha \log n)$ bits of communication. 
Node $u$ then computes 
\begin{align*}
\Psi_u(i_0)
&= \sum_{t\in [\alpha]}\sum_{v\in S_u} \Psi_{uv,t}(i_0) 
= \sum_{t\in [\alpha]}\sum_{v\in S_u} \Psi_{S_u,t}(i_0)\cdot\Psi_{S_v,t}(i_0)
\end{align*}
and accepts if and only if $\Phi_u(i_0)=\Psi_u(i_0)$. 
The probability that 
two non-equal polynomials on $\fq$ of degree at most $2(\frac{n}{\alpha}-1)$
are equal at a random point $i$ is at most $2(\frac{n}{\alpha}-1)/q$. Therefore, since $q>cn\alpha$, the probability of error can be made arbitrarily small by choosing $c$ large enough.
\end{proof}

\noindent\textit{Remark.} Similar trade-offs can still be obtained even if nodes have only access to distributed randomness. For instance, in two  rounds, with the same notations as in the proof of Theorem~\ref{theo:triaglefreeness}, we can have each node $u$ choose its own random $i_u\in\fq$, and send it to all its neighbors $v$. 
To get a 1-round $\MA$ protocol,  
with $O(\frac{n^2}{\alpha}\log n)$-bit certificates, and $O(\alpha \log n)$-bit communication, Merlin sends to $u$ a specific certificate for each edge incident to $u$. 
That is, $u$ gets a polynomial $\Phi_v$ for each $v\in S_u$,
which equals (allegedly) to 
$
\Psi_{uv}(i)
= \sum_{t\in T} \Psi_{uv,t}(i)
= \sum_{t\in T} \Psi_{S_u,t}(i)\cdot\Psi_{S_v,t}(i)
$
on each $i\in\fq$. 
In this case, $v$ chooses $i_v$ at random locally,
and sends to $u$ the value $i_v$ in addition to the $\alpha$ evaluations $\Psi_{S_v,t}(i_v)$ for all $t\in[\alpha]$.

A particular application of Theorem~\ref{theo:triaglefreeness} is the existence of a Merlin-Arthur protocol with both space and message complexities $\widetilde{O}(\sqrt{n})$. This contrasts with the following lower bound.

\begin{proposition}
\label{prop: lb for triangle-freeness}
Any proof-labeling scheme for triangle-freeness must have certificate size at least $n/e^{O(\sqrt{\log n})}$ bits.
\end{proposition}

\begin{proofidea}
The lower bound graph construction for the \textsc{broadcast-congested-clique} model~\cite{DruckerKO14} obviously gives a lower bound to the weaker, \textsc{broadcast-congest} model. This lower bound is based on a lower bound for  multiparty communication complexity of disjointness~\cite{RY15}, which also applies for the non-deterministic case. Finally, as noted in previous work on PLS~\cite{GS16,Censor-HillelPP17}, a lower bound for non-deterministic communication complexity in the \textsc{broadcast-congest} model implies a certificate-size lower bound for PLS.
\end{proofidea}
% %-------------------------------------
% \paragraph*{A lower bound for triangle freeness.}
% %-------------------------------------

We can then conclude that, as opposed to the $\MA$ protocol of Theorem~\ref{theo:triaglefreeness},  any PLS for triangle freeness must use almost-linear communication. Put differently, the trivial protocol of sending all the list of neighbors is almost optimal, even if non-determinism is used.

%%%%%%%%%%%%%%%%%%%%%%%%%%%%%%%%%%%%%%%%%%%
\section{Distributed vs. Shared Randomness}
%%%%%%%%%%%%%%%%%%%%%%%%%%%%%%%%%%%%%%%%%%%

In this section we compare the power of distributed interactive protocols using \emph{shared} randomness (the nodes have access to a common source of random coins) with the power of protocols using \emph{distributed} randomness (each node has access to a private source of random coins only) --- in both cases, the outcomes of the random trials are public to Merlin. 

\subsection{Interactive Protocols that use Shared Randomness}

%-------------------------------------
\paragraph*{Certifying Solutions to Optimization Problems.}
%-------------------------------------

We consider optimization problems on graphs, such as finding a minimum dominating set, or a maximum independent set, and their weighted counterparts. Similar problems where previous studied in the context of non-interactive distributed verification~\cite{FFH16}.
In such a problem~$\pi$, an admissible solution is a set $S$ of nodes satisfying a set of constraints depending on~$\pi$, and the quality of a solution $S$ is measured by its weight $w(S)=\sum_{s\in S}w(s)$ where $w(s)$ is the weight of node $s$, given as input (where $w(s)=1$ for every node $s$ when considering only the cardinality of the solution). 
We assume that all weights are polynomial in the size~$n$ of the network. A set $S$ is distributively encoded by a Boolean variable~$x(v)$ at each node~$v$, indicating whether the node is in~$S$ or not. 
We consider two distributed languages: 
\begin{itemize}
\item The language $\adm_\pi$ is composed of all configurations $(G,(w,x),\id)$ such that $x$ encodes an \emph{admissible} solution for $\pi$ in the weighted graph $G$ (weights are assigned by $w$). 
\item The language $\optim_{\pi,k}$, for $k\geq 0, $ is composed of all configurations $(G,w,\id)$ such that there exists an admissible solution for $\pi$ of weight at most~$k$ (respectively, at least~$k$) for the minimization (respectively, maximization) problem~$\pi$. 
\end{itemize}
This framework can easily be extended to study problems whose solutions are sets of edges.

\begin{theorem}\label{theo:optimization}
For any optimization problem $\pi$ on graphs such that checking whether a given solution $x$ is admissible can done by exchanging $O(\log\log n)$ bits between neighbors, there exists a Merlin-Arthur protocol for $\optim_{\pi,k}$, using $O(\log n)$ bits of shared randomness, with $O(\log n)$-bit certificates and $O(\log\log n)$-bit messages between nodes. In short, $\optim_{\pi,k}\in\MA(\log n,\; \log\log n)$.
\end{theorem}

\begin{proof}
Let $\pi$ be an optimization problem, and assume, w.l.o.g., that $\pi$ is a minimization problem (e.g., minimum dominating set). Let us consider a legal configuration $(G,w,\id)$, i.e., there exists $S\subseteq V$ such that $w(S)\leq k$.  Merlin assigns a Boolean $x(v)$ to every node~$v$, stating whether $v$ is in $S$ or not. Note that, by hypothesis, the admissibility of $S$ is checkable by exchanging only $O(\log\log n)$ bits between neighbors. (Just one bit suffices in the case of minimum dominating set.) In order to measure the quality of the solution, Merlin also provides the nodes with a distributed encoding of a  tree $T$ spanning $G$ rooted at an arbitrary node~$r$ (i.e., provide every node $v$ with the identifier of its parent $p(v)$), and a distributed proof that $T$ is a spanning tree. Again, it is known (see~\cite{KKP10}) that certificates of $O(\log n)$ bits suffice for certifying such a tree (by using, e.g., $(\id(r),cpt(v))$ where $cpt(v)$ is the distance to~$r$ in~$T$). Every node $u$ is also given the partial sum $s(u)=\sum_{v\in V(T_u)\cap S}w(v)$ as part of its certificate, where $T_u$ denotes the subtree of $T$ rooted at~$u$. Each certificate consumes $O(\log n)$ bits at each node. 

Verifying at every node $u$ that (1)~the identifier of $r$ given to $u$ is the same as the one given to its neighbors in $G$, (2)~its distance $cpt(u)$ to $r$ in $T$ is one more than the one of its parent in $T$, and one less than the one of its children in $T$, and (3)~$s(u)=\sum_{p(v)=u}s(v) +x(u)\cdot w(u)$, can trivially be done by exchanging $O(\log n)$ bits between neighbors. We show how, using shared randomness, we can reduce this to $O(\log\log n)$ bits, without increasing the certificate size.

The tests~(1) and~(2) are equality tests, which can be done by exchanging only $O(1)$ bits between neighbors with access to $O(\log n)$ shared random bits~\cite{KN97}. The test~(3) consists of checking that $x(u)\cdot w(u)-s(u)+\sum_{p(v)=u}s(v)=0$ which is known as the $\seq$ problem in multi-party communication complexity. The protocol in~\cite{N93} achieves this test with $O(\log\log n)$-bit communication complexity, using $O(\log\log n)$ shared random bits. 
\end{proof}

For Theorem~\ref{theo:optimization}, we assumed that all weights are polynomial in the size~$n$ of the network. If the weights are $m$-bit long, we can adapt the proof, and show that there exists a Merlin-Arthur protocol for $\optim_{\pi,k}$, using $O(\log(m + \log n))$ bits of shared randomness, with $O(\log n)$-bit certificates and $O(\log n)$-bit messages between nodes. In short,  $\optim_{\pi,k}\in\MA(\log n,\; \log n)$ even with weights exponential in~$n$.

\paragraph*{Certifying Coloring and Lucky Labeling.}
Similar arguments as the ones used to establish Theorem~\ref{theo:optimization} allow us to verify specific optimization problems, for which checking that a solution is admissible is not easy. We exemplify this with the coloring problem, and its variant the \emph{lucky labeling} problem~\cite{AHPSV16,CGZ09}.

Checking that a given graph coloring is proper is a simple task,
which can be solved by having each node broadcast its color.
Here, we show that verifying a given $c$-coloring 
can be done using a $\MA$ protocol with $O(\Delta\log\log c)$-bit certificates and $O(1)$ bits of communication
in networks of maximum degree~$\Delta$. 
This stands in contrast to the trivial verification algorithm where the communication is of $O(\log c)$ bits.
In the $\MA$ protocol, Merlin provides every node $v$ with 
the location $p(v,u)$ of a bit where the colors of $u$ and $v$ differ, for each of its neighbors $u$.
A node $v$ then checks with each neighbor $u$ the fact that $p(v,u)=p(u,v)$, and at the same time sends to $u$ the  value of the corresponding bit.
The Merlin step requires certificates of $O(\Delta \log\log c)$ bits. The Arthur step requires constant communication when shared randomness is available, using the equality protocol.

\begin{lemma}\label{lem:MA for coloring}
There is a Merlin-Arthur protocol for verifying a given $c$-coloring using $O(\log\log c)$ bits of shared randomness, certificates of size $O(\Delta\log\log c)$ bits, and constant communication complexity. 
In short, $c$-coloring of graphs of degree at most $\Delta$ is in  $\MA(\Delta\log\log c,\;1)$.
\end{lemma}

We apply this observation to the so-called \emph{lucky labeling}. For a graph $G=(V,E)$ and a positive integer $c$, let $\ell:V\to \{1,\dots,c\}$, and, for every node $v$ of $G$, let $S(v)=\sum_{u\in N(v)}\ell(u)$. The labeling $\ell$ is \emph{lucky} if, for every two adjacent nodes $u$ and $v$, we have $S(u)\neq S(v)$. The lucky coloring number of a graph $G$, denoted by $\eta(G)$, is the least positive integer $c$ such that $G$ has a lucky labeling $\ell:V\to \{1,\dots,c\}$. We refer to~\cite{AHPSV16,CGZ09} for properties of the lucky coloring number. In particular, it is conjectured that $\eta(G)\leq \chi(G)$, and it is known that  $\eta(G)\leq \Delta^2-\Delta+1$, even for list lucky labeling. 

Verifying a given graph coloring is trivial, 
even without labels or interaction. 
To verify an upper bound on the chromatic number $\chi$ of a graph, there is a simple PLS giving each node a color.
The situation with lucky labeling is much more subtle: it is impossible to verify lucky labeling in a single round 
(this can be easily seen by considering different labelings on a short path).
There is a simple PLS for verifying a given labeling is lucky, or for bounding $\eta$ from above,
which gives each node $v$ the sum $S(v)$, and also labels for the latter case.

This PLS has label size and communication of $O(\log\Delta)$.
Applying RPLS~\cite{BFPS15} gives $\sigma = O(\Delta \log \Delta)$ and $\gamma= O(\log\log\Delta)$, which can be reduced to $\gamma= O(1)$ using shared randomness. 
Here, we show an MA protocol using shared randomness
with $\sigma=O(\Delta\log\log\Delta)$ 
and $\gamma=O(\log\log\Delta)$, establishing another trade-off between space and communication.

\begin{theorem}
For every $\lambda = O(n)$, 
there exists a Merlin-Arthur protocol for $\eta(G)\leq\lambda$, using $O(\log\log\Delta)$ bits of shared randomness, with $O(\Delta\log\log\Delta)$-bit certificates and $O(\log\log\Delta)$-bit messages between nodes. In short, lucky labeling is in  $\MA(\Delta\log\log\Delta,\; \log\log\Delta)$.  
\end{theorem}

\begin{proof}
Merlin sends to every node $v$ its label $\ell(v)$ and the alleged sum $S(v)$ of the labels of its neighbors.
For the verification, the nodes verify that the sums $S(v)$ constitute a proper coloring,
using the protocol from Lemma~\ref{lem:MA for coloring}. 
In addition, they use a multiparty protocol for the \seq{} problem~\cite{N93,AKM13},
in order to verify $S(v)=\sum_{u\sim v}\ell(u)$.
\end{proof}

Note that a similar protocol can be used for the problem of verifying that a given labeling is lucky.

%-------------------------------------
\subsection{A General Reduction Between Distributed and Shared Randomness}
%-------------------------------------

Interestingly, assuming distributed randomness does not limit the power of Arthur-Merlin protocols compared to \emph{shared} randomness, up to a small additive factor in the certificate size. The same holds for Merlin-Arthur protocols, but solely up to one additional interaction between Arthur and Merlin. 
This result is not hard to achieve using a classical spanning-tree verification technique already applied in proof  labeling schemes~\cite{KKP10}. Yet, it both generalizes and simplifies the previous results on shared vs.\ distributed randomness~\cite{KOS18}, so we present it here in full.

\begin{theorem}\label{theo:general-reduction}
For any distributed language $\cL$, and for any number $k\geq 1$ of interactions, and for any certificate size $\sigma\geq 0$, if $\cL\in\AM[k](\sigma,\gamma)$ (respectively, $\cL\in\MA[k](\sigma,\gamma)$) using $\rho(n)$ shared random bits, then $\cL\in\AM[k](\sigma+\log n + \rho,\gamma+\log n + \rho)$ (respectively, $\cL\in\AM[k+1](\sigma+\log n + \rho,\gamma+\log n + \rho)$) with distributed randomness.
\end{theorem}

\begin{proof}
Let $\cL\in\AM[k](\sigma,\gamma)$ with an interactive Arthur-Merlin protocol $\P$ using $\rho(n)$ shared random bits. If nodes have only access to distributed random coins, then they can simulate $\P$ as follows. At every interaction with Merlin, every node sends a random string of $\rho(n)$ bits to Merlin, and Merlin is bound to choose the one generated by the node with minimum identifier. To prove that Merlin does so, it provides every node $v$ with a distributed certificate $c(v)$ for proving a spanning tree $T$ of the network, rooted at the node with minimum identifier. As we already mentioned previously in the paper, it is known that certificates of $O(\log n)$ bits suffice for certifying such a tree~\cite{KKP10}. All nodes are also provided with the random string $r$ produced by the node with smaller identifier, which consume $\rho(n)$ bits of certificates. All nodes check that they are given the same random bits, and the root of the tree checks that these bits are those it gave to Merlin.  

For the case of $\cL\in\MA[k](\sigma,\gamma)$, the proof is identical, except for the last stage of the interactive Merlin-Arthur protocol $\P$ verifying $\cL$, which involves a distributed verification algorithm using shared randomness. This latter algorithm can be simulated in the distributed randomness setting by adding another interaction with Merlin, which enables to certify the random string generated by the node with minimum identifier, to be used instead of the shared random string originally used by the distributed verification algorithm.
\end{proof}

%-------------------------------------
\subsection{Lower Bounds for Shared Randomness}
%-------------------------------------

For many verification problems, the number of random bits used by Arthur remains limited, typically $\rho(n)=O(\log n)$, which shows that, often, shared randomness does not add much power to Arthur-Merlin protocols.
The next result states a lower bound on the certificate and message size in the case of $\sym$ and $\nsym$, even when using shared randomness.

\begin{theorem}
\label{thm:lowerboundsymmetry}
Any Arthur-Merlin protocol for (non) symmetry must have certificate and message size $\Omega(\log\log n)$. In short, $\sym,\nsym\not\in\AM(o(\log\log n), \infty) \cup\AM(\infty, o(\log\log n))$,
even using shared randomness.
\end{theorem}

% \begin{proofidea}
% We give a general reduction to (non-distributed) Arthur-Merlin games~\cite{GPW16}, which simultaneously yields lower bounds for $\sym$ and for $\nsym$ in the distributed case, and both for space complexity and for communication complexity.
% \end{proofidea}

\begin{proof}
In~\cite{GPW16}, a communication complexity variant of Arthur-Merlin protocols has been proposed. In this variant, Arthur consists of two parties, Alice and Bob, and the input is split between them: Alice holds $x$, Bob holds $y$, and they wish to decide whether the value of a specified function $f$ with input $x$ and $y$ is equal to 1. At the beginning, Alice and Bob start by tossing some coins, then Merlin publishes a certificate, and finally Alice and Bob separately decide whether to accept (the acceptance/rejection criteria are the same as for the Arthur-Merlin protocols). The communication cost of the protocol is defined as the worst-case length of Merlin’s certificates. At the end of the paper, the authors observe that, with respect to this variant of Arthur-Merlin protocols, any such protocol for $\equ$ and for $\nequ$ must have communication complexity $\Omega(\log\log n)$. We will now show that the existence of a \AM\ protocol with one interaction for the $\sym$ (respectively, $\nsym$) problem with certificate size $o(\log\log n)$ would imply a two-player Arthur-Merlin protocol for $\equ$ (respectively, $\nequ$) with communication complexity $o(\log\log n)$. The theorem thus follows.

Given two binary vectors $x=(x_1,\ldots,x_n)$ and $y=(y_1,\ldots,y_n)$, recall that $\equ(x,y)=1$ if and only if $x_i=y_i$ for every $i$ with $1\leq i\leq n$ (for the sake of simplicity, we will assume that $x\neq \mathbf{0}$ and that $y\neq \mathbf{0}$, but our construction can be also adapted to the case in which $x= \mathbf{0}$ or $y= \mathbf{0}$). We now define a graph $G_{x,y}$ such that $\sym(G_{x,y})=1$ if and only if $\equ(x,y)=1$ (and, hence, $\nsym(G_{x,y})=1$ if and only if $\nequ(x,y)=1$). The graph includes $6n+2$ nodes $a$, $b$, $a_i$, $b_i$, $u_i$, $v_i$, $y_i$, and $z_i$, for $1\leq i\leq n$, and the following edges (see Fig.~\ref{fig:lowerboundsymmetry}):

\begin{itemize}
\item $(a,b)$, $(a,a_i)$, for $1\leq i\leq n$ such that $x_i=1$, and $(b,b_i)$, for $1\leq i\leq n$ such that $y_i=1$;

\item $(u_i,a_j)$ and $(v_i,b_j)$, for $1\leq i\leq j\leq n$;

\item $(u_i,u_j)$ and $(v_i,v_j)$, for $1\leq i < j \leq n$, and $(u_i,y_j)$ and $(v_i,z_j)$, for $1\leq i, j \leq n$;

\item $(y_i,y_{i+1})$ and $(z_i,z_{i+1})$, for $1\leq i < n$, and $(y_1,y_n)$ and $(z_1,z_n)$.
\end{itemize}
Clearly, if $x=y$, then $\sym(G_{x,y})=1$: indeed, we can simply map each $a$-node (respectively, $u$-node and $y$-node) to the corresponding $b$-node (respectively, $v$-node and $z$-node). On the other hand, because of the degree distribution of its nodes, any non-trivial automorphism of $G_{x,y}$ has to map the $u$-nodes to the corresponding $v$-nodes: this in turn implies that, because of their neighborhoods, each $a$ node has to be mapped to the corresponding $b$-node. Hence, since the mapping is an automorphism, the neighborhood of node $a$ and node $b$ has to be the same: that is, $x=y$.

%%%%%%%%%%%%%%%%%%%%%%%%%%%%%%%%%%%%%%%%%%%%%%%%%%%%%%
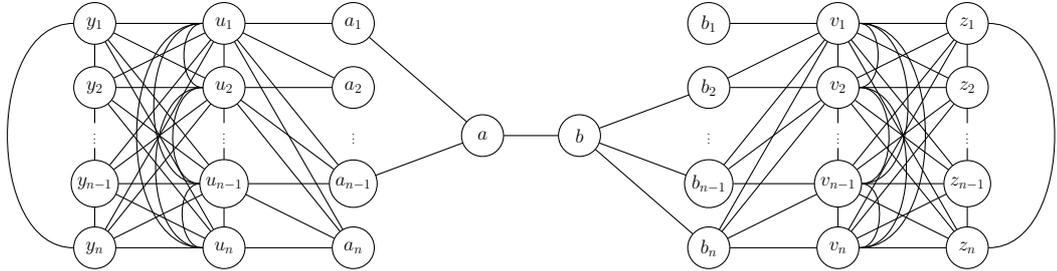
\begin{figure}[htb]
\begin{tikzpicture}[scale=0.425,every node/.style={scale=0.425,minimum width=1.3cm}]
		\foreach \xsc/\xsh/\nodea/\nodeu/\nodey in {1/0/a/u/y,-1/-11/b/v/z} {
			\begin{scope}[xscale=\xsc,xshift=\xsh cm]
				\node[circle,draw] (\nodea) at (4,-3.5) {\LARGE{$\nodea$}};
				\foreach \x/\y/\z in {1/0/1,2/-2/2,3/-5/{n-1},4/-7/n} {
					\node[circle,draw] (\nodea\x) at (0,\y) {\LARGE{$\nodea_{\z}$}};
					\node[circle,draw] (\nodeu\x) at (-4,\y) {\LARGE{$\nodeu_{\z}$}};
					\node[circle,draw] (\nodey\x) at (-8,\y) {\LARGE{$\nodey_{\z}$}};
				}
				\node (\nodea5) at (0,-3.5) {$\vdots$};
				\node (\nodeu5) at (-4,-3.5) {$\vdots$};
				\node (\nodey5) at (-8,-3.5) {$\vdots$};
				% Edges a-u
				\foreach \x/\y  in {1/1,1/2,1/3,1/4,2/2,2/3,2/4,3/3,3/4,4/4} {
					\draw (\nodeu\x) -- (\nodea\y);
				}
				% Vertical edges
				\foreach \x/\y  in {1/2,2/5,5/3,3/4} {
					\draw (\nodeu\x) -- (\nodeu\y);
					\draw (\nodey\x) -- (\nodey\y);
				}
				% Edges u-y
				\foreach \x  in {1,2,3,4} {
					\foreach \y  in {1,2,3,4} {
						\draw (\nodeu\x) -- (\nodey\y);
					}
				}
				% Bended edges
				\draw (\nodey1) to[bend right=90] (\nodey4);
				\foreach \x/\y  in {1/2,1/3,1/4,2/3,2/4,3/4} {
					\draw (\nodeu\x) to[bend right=90] (\nodeu\y);
				}
			\end{scope}
		}
		% Input encoding
		\draw (a) -- (a1);
		\draw (a) -- (a3);
		\draw (b) -- (b2);
		\draw (b) -- (b3);
		\draw (b) -- (b4);
		\draw (a) -- (b);
	\end{tikzpicture}
\caption{The graph used to reduce $\equ$ (respectively, $\nequ$) to $\sym$ (respectively, $\nsym$): in this case, $x_2=x_n=y_1=0$}
\label{fig:lowerboundsymmetry}
\end{figure}
%%%%%%%%%%%%%%%%%%%%%%%%%%%%%%%%%%%%%%%%%%%%%%%%%%%%%%

Let us now suppose that there exists a \AM\ protocol $\P$ with one interaction for the $\sym$ (respectively, $\nsym$) problem which uses certificates of size $o(\log\log n)$. We now show how $\P$ can be used to design an Arthur-Merlin protocol for $\equ$ (respectively, $\nequ$) with communication complexity $o(\log\log n)$ (for the sake of brevity, we will show this statement for $\sym$ and $\equ$: the proof for $\nsym$ and $\nequ$ is almost identical). Given $x$ (respectively, $y$), Alice (respectively, Bob) can construct the $(a,u,y)$-subgraph (respectively, $(b,v,z)$-subgraph) of $G_{x,y}$: let $G_{x,y}^A$ (respectively, $G_{x,y}^B$) denote such subgraph. After having sent to Merlin the shared random string $r$, Alice and Bob waits for Merlin's certificate which is supposed to be formed by the two certificates $\pi_a$ and $\pi_b$ that nodes $a$ and $b$ would have received during the execution of $\P$ with random string $r$. By simulating $\P$ for every possible certificate assignment to the nodes of $G_{x,y}^A$ (respectively, $G_{x,y}^B$), Alice (respectively, Bob) can verify whether there exists an assignment that makes all the nodes of its corresponding subgraph accept: if this is the case, Alice (respectively, Bob) accepts. By definition, we have that if $x=y$, then, for any random string $r$ there exist a certificate assignment to $G_{x,y}^A$ (respectively, $G_{x,y}^B$) and a certificate for Alice and Bob which make Alice and Bob accept. On the other hand, if $x\neq y$, then, for at least \nicefrac{2}{3} of all possible random strings, any certificate assignment to $G_{x,y}^A$ (respectively, $G_{x,y}^B$) and any certificate for Alice and Bob makes Alice and Bob reject. Since the size of the certificate for Alice and Bob is twice the size of the certificate size of $\P$, this implies that this protocol is an Arthur-Merlin protocol for $\equ$ with communication complexity $o(\log\log n)$. This contradicts the lower bound observed in~\cite{GPW16}: we have thus proved that $\sym,\nsym\not\in\AM(o(\log\log n), \infty)$.

The above proof can be adapted in order to obtain a lower bound on the communication complexity of any \AM\ for the $\sym$ and $\nsym$ problems. Indeed, instead of asking Merlin for the certificates of the nodes $a$ and $b$, Alice and Bob ask Merlin for the message transmitted on the edge $(a,b)$. They then try to find a certificate assignment that suits this message. Hence, we have also shown that $\sym,\nsym\not\in\AM(\infty, o(\log\log n))$ and the theorem follows.
\end{proof}

A result similar to previous theorem was proved in~\cite{KOS18} in an ad-hoc manner, but only for the $\sym$ problem and with respect to space complexity. The authors have recently reported to improve the lower bound from $\log \log n$ to $\log n$~\cite{NPY18}.

%%%%%%%%%%%%%%%%%%%%%%%%%%%%%%%%%%%%%%%%%%%
\section{Interactions vs. Space and Communication }
%%%%%%%%%%%%%%%%%%%%%%%%%%%%%%%%%%%%%%%%%%%

In this section, we explore the power given to interactive protocols by allowing many interactions between Merlin and Arthur, in terms of both space complexity and communication complexity. 

%-------------------------------------
\subsection{Reducing the Number of Interactions}
%-------------------------------------
The following general result allows us to reduce the number of interactions between Arthur and Merlin, at the cost of increasing the certificate size and the communication cost of the protocol.

\begin{theorem}
\label{thm:reduction}
 For any two functions $\sigma$ and $\gamma$, $\MAM(\sigma,\gamma) \subseteq \AM(n\sigma^2,n\sigma \gamma)$.
 \end{theorem}

\begin{proof}
Let $\P$ be a $\MAM(\sigma,\gamma)$ 1-sided protocol for a language $\cL$ using $\rho(n)$ shared random bits. Given a configuration $I=(G,\id,x)$, let $y^I_1 : V \rightarrow \Sigma^\sigma$ be the function specifying the certificate sent by Merlin to each node during the first iteration. Moreover, for any string $r\in \Sigma^{\rho(n)}$, let $y^{I,r}_2 : V \rightarrow \Sigma^\sigma$ be the function specifying the certificate sent by Merlin to each node during the second iteration, whenever $r$ is the shared randomly chosen string by Arthur. From the definition, it follows that, if $I\not\in \cL$, then $\Pr[\P \; \mbox{accepts at all nodes}]\leq \nicefrac13$, that is, for at most $\frac{2^{\rho}}{3}$ random strings $r$ there exists a function $y^{I,r}_2$ which makes the protocol $\P$ accept. By repeating $k$ times the protocol $\P$ on the configuration $I$, we then have that, for at most $\frac{2^{k\rho}}{3^k}$ $k$-tuples of random strings $r_1,\ldots,r_k$, there exists a $k$-tuple of functions $y^{I,r_1}_2,\ldots,y^{I,r_k}_2$ which makes the protocol $\P$ accept at each repetition. Let $\BAD^I(y^I_1)$ denote the set of such $k$-tuples of random strings: we have that $|\BAD^I(y^I_1)|\leq \frac{2^{k\rho}}{3^k}$.

We can now define a $\AM$ protocol $\P'$ in the following way. Arthur chooses $k$ shared random strings $r_1,\ldots,r_k$ and sends $R=r_1\cdots r_k$ to Merlin. Merlin sends to each node a label formed by $k+1$ strings in $\Sigma^\sigma$ chosen according to $k+1$ functions $y^{I,R}_1,y^{I,R}_{2,1},\ldots,y^{I,R}_{2,k} : V \rightarrow \Sigma^{\sigma}$. At this point, each node $u$ executes $k$ repetitions of protocol $\P$, assuming that the first certificate received by Merlin is $y^{I,R}_1(u)$, and that the second certificate received from Merlin at the $i$-th repetition is $y^{I,R}_{2,i}(u)$: $u$ accepts if and only if all repetitions accept. If $I\in \cL$, then all nodes will accept with probability one, since $\P$ is a 1-sided protocol. On the other hand, if $I\not\in \cL$, then the probability that $\P'$ accepts can be bounded by referring to the cardinality of $\bigcup_{y^I_1}\BAD^I(y^I_1)$, where the union is over all possible functions $y^I_1 : V \rightarrow \Sigma^\sigma$: indeed, given the function $y^{I,R}_1=y^I_1$, this cardinality is an upper bound on the number of random strings $R=r_1\cdots r_k$ such that there exist $k$ functions $y^{I,R}_{2,1},\ldots,y^{I,R}_{2,k}$ which allows Merlin to make all nodes accept.
 
 From the bound above and from the fact that we have exactly $2^{n\sigma}$ functions mapping $n$ nodes to strings of length $\sigma$, we have that
 
 \[
 |\bigcup_{y^I_1}\BAD^I(y^I_1)|\leq \sum_{y^I_1}\left|\BAD^I(y^I_1)\right|\leq \sum_{y^I_1}\frac{2^{k\rho}}{3^k}=2^{n\sigma}\frac{2^{k\rho}}{3^k}.
 \]
 Hence,
 \[
 \Pr[\P \; \mbox{accepts at all nodes}]\leq \frac{2^{n\sigma}\frac{2^{k\rho}}{3^k}}{2^{k\rho}} = \frac{2^{n\sigma}}{3^k}.
 \]
 If we choose $k = n\sigma$, the above probability is at most $\nicefrac{1}{3}$. On the other hand, with $k = n\sigma$ we have that $\P'$ uses certificates of size $(n\sigma+1)\sigma$ and has a communication complexity equal to $n\sigma\gamma$ (since it has to simulate $n\sigma$ repetitions of protocol $\P$ which has communication cost equal to $\gamma$). The theorem is thus proved. 
\end{proof}
 
 \begin{corollary}\label{cor:reduction}
 For any two functions $\sigma$ and $\gamma$, $\AMAM(\sigma,\gamma) \subseteq \AM(n\sigma^2,n\sigma \gamma)$.
 \end{corollary}

As a direct application of Theorem~\ref{thm:reduction}, we have that $\sym$ admits a \AM\ protocol with one interaction and certificate size $O(n\log^2 n)$. This is a consequence of the theorem and of the existence of a \AM\ protocol with two interactions~\cite{KOS18}. It is worth noting that this consequence of our general reduction result is only a $\log$ factor away from the ``ad-hoc'' result of~\cite{KOS18} which establishes the existence of a \AM\ protocol with one interaction and certificate size $O(n\log n)$.
Corollary~\ref{cor:reduction} can be applied to a more recent result\cite{NPY18}, which establishes the existence of a \AM\ protocol with three interactions and certificate size $O(\log n)$ for the non-isomorphism graph problem, in the case in which the nodes can communicate on both graphs. As a consequence of the corollary, we have that this problem admits a \AM\ protocol with one interaction and certificate size $O(n\log^2 n)$. 
As far as we know, this is the first \AM\ protocol with one interaction for this version of the non-isomorphism graph problem An interesting open question is whether such a protocol can exist also for the problem in which nodes can communicate only on one graph, while the other graph is locally given as input to the nodes themselves. For this latter problem, a \AM\ protocol with one interaction and certificate size $O(n\log n)$ was given~\cite{KOS18}, as well as an \AM\ protocol with a constant number of interactions and certificate size $O(\log n)$~\cite{NPY18}. Observe that the two applications of Theorem~\ref{thm:reduction} and of Corollary~\ref{cor:reduction} are obtained at the cost of an increase of the communication complexity by a factor $\widetilde{O}(n)$. We do not know if this linear increase of communication complexity can be avoided in general. 

%-------------------------------------
\subsection{The Arthur-Merlin Hierarchy}
%-------------------------------------

We analyze the power of the Arthur-Merlin hierarchy. Recall that, for any $\sigma\geq 0$ and $\gamma\geq 0$, $\AMH(\sigma,\gamma)=\bigcup_{k\geq 0}\AM[k](\sigma,\gamma)$.
% where $k$ denotes the number of interactions between the prover and the verifier, $\sigma$ denotes the maximum size of the certificates provided by Merlin, and $\gamma$ denotes the maximum size of the messages exchanged between neighbors in the network. 
We show that increasing the number of interactions cannot help much for reducing the certificate size to $o(n)$, even for languages defined on a very simple subclass of regular graphs, with 1-bit inputs, and admitting a locally checkable proof with $O(n)$-bit certificates.   

\begin{theorem}\label{thm:outsideAM}
There exists a distributed language $\cL$ on cycles, with 1-bit inputs, admitting a locally checkable proof with $O(n)$-bit certificates, and $O(n)$-bit messages, that is outside the Arthur-Merlin hierarchy with $o(n)$-bit certificates, even with messages of unbounded size, and even if the verifier performs an arbitrarily large constant number of rounds, whenever Arthur generates $\rho(n)=o(n)$ random bits at each node for each interaction with Merlin. In short, there exists a distributed language $\cL$ on regular graphs satisfying 
$
\cL\in\Sigma_1\LD(O(n),O(n))\setminus \AMH(o(n),\infty).
$
\end{theorem}

\begin{proof}
We  show that there exists a language on 0/1-labelled oriented cycles that is outside $\AMH(o(n),\infty)$, where an oriented cycles is a cycle $C_n=(u_0,\dots,u_{n-1})$ in which the nodes are provided with IDs in $\{0,\dots,n-1\}$, given consecutively to the nodes, i.e., $\id(u_i)=i$. (Node $u_i$ is adjacent to nodes $u_{i+1 \bmod n}$ and $u_{i-1 \bmod n}$). In addition, 0/1-labelled oriented cycles means that we restrict ourselves to languages with binary inputs. The proof is based on a construction and a counting argument similar to the ones used in~\cite{FFH16} for proving that there are languages outside the local hierarchy $\LH$. 

Let $k$ and $t$ be two non-negative integers, and let $\sigma:\mathbb{N}\to\mathbb{N}$ with $\sigma\in o(n)$. First, we show that there exists an integer $n$, and a language on 0/1-labelled oriented cycles that cannot be recognized by a protocol in $\AM[k](\sigma,\infty)$ with Arthur generating $r=o(n)$ random bits at each interaction, and running in $t$-round verifier. Observe that the verifier at every node $u_i$ takes as input the $(2t+1)k\cdot \sigma(n)$ bits provided by Merlin to the nodes in the $t$-ball around $u_i$ during the $k$ interactions between Arthur and Merlin at these nodes, the $k\cdot \rho(n)$ random bits generated by Arthur, and the identifier~$i$ of $u_i$. The number of functions 
\[
\mu: \{0,1\}^{(2t+1)k\cdot \sigma(n)+k\cdot \rho(n)} \to \{\acc,\rej\}
\]
is at most 
\[
2^{2^{(2t+1)k\cdot \sigma(n)+k\cdot \rho(n)}}.
\]
It follows that the number of interactive protocols at $u_i$ with Arthur generating $\rho(n)$ random bits at each interaction, and running a $t$-round verifier, with Merlin using certificates of size at most $\sigma(n)$ bits is at most $2^{2^{(2t+1)k\cdot \sigma(n)+k\cdot \rho(n)}}$. Since the nodes have IDs, the total number of protocols is at most $2^{n2^{(2t+1)k\cdot \sigma(n)+k\cdot \rho(n)}}$.

On the other hand, the number of languages on binary strings of length $n$ is exactly $2^{2^n}$. So, let $n_0$ be such that 
\[
2^{n_0 2^{(2t+1)k\cdot \sigma(n_0)+k\cdot r(n_0)}}<2^{2^{n_0}}.
\]
Such an $n_0$ exists since $\sigma(n)=o(n)$ and $\rho(n)=o(n)$, and $k$ and $t$ are constant.  By the pigeon-hole principle, there exists a language on the 0/1-labeled oriented cycle $C_{n_0}$ that cannot be decided by any protocol in $\AM[k](\sigma(n_0),\infty)$ with Arthur generating $r(n_0)$ random bits at each interaction, and running a $t$-round verifier. Indeed, on the $n_0$-node oriented cycle, the verifier acts exactly the same for at least two languages $\cL$ and $\cL'$, and thus is incorrect for at least one of these two languages. 

Let $m$ be a nonnegative integer, and let $S(n,m)$ be the set of languages on 0/1-labeled oriented cycles with $n$ nodes that cannot be recognized by a protocol with $k=t=m$ (i.e., $m$ interactions between Arthur and Merlin, and $m$ rounds of communication between the nodes). It follows from the first part of the proof that, for every $m$, there exists $n$ such that $S(n,m)\neq\emptyset$. 
So, let 
\[
 n_{min}(m)=\min\{n\geq 1 : S(n,m)\neq\emptyset\}.
\]
Also, let
\[
m_{max}(n)=\max\{m\geq 0 : S(n,m)\neq\emptyset\}.
\] 
Note that $m_{max}$ is well defined since, for every $n\geq 3$, there are languages on 0/1-labeled $C_n$ that cannot be decided in zero rounds without interactions with a prover (corresponding to $m=0$), and every language on $C_n$ is decidable in $\lceil n/2\rceil$ rounds (corresponding to $m=\lceil n/2\rceil$). Note also that, for every $n\geq 3$ and $m\geq 0$, we have
\[
m_{max}(n_{min}(m)) \geq m.
\]
Given $n$ and $m$ such that $S(n,m)\neq \emptyset$, let $\cL(n,m)$ be the smallest language of $S(n,m)$ according to the lexicographic ordering of the collections of binary words of $n$ bits corresponding to the inputs in the language. Let 
\[
\cL= \bigcup_{n\geq 3} L(n,m_{max}(n)).
\]
Note that $L$ is indeed a distributed language, for it is Turing-computable. We show that 
\[
\cL\notin \AMH(\sigma,\infty).
\]
Suppose, for the sake of contradiction, that $\cL\in \AMH(\sigma,\infty)$. Then there exists an interactive protocol~$\mathcal{P}$ in $\AM[k](\sigma,\infty)$ for $\cL$, with Arthur tossing $\rho(n)$ bits at each interaction, and performing a $t$-round verification algorithm, for some $k\geq 0$ and $t\geq 0$. Let  $m=\max\{k,t\}$.  We have $\cL\in\AM[m](\sigma,\infty)$. Let us then consider the restriction $\cL'$ of $\cL$ on the oriented cycle with $n_{min}(m)$ nodes, that is, 
\[
\cL'=\cL(n_{min}(m), m_{max}(n_{min}(m))).
\] 
Since 
$
\cL'\in S(n_{min}(m), m_{max}(n_{min}(m))),
$
it cannot be recognized by an interactive $\AM[m'](\sigma,\infty)$ protocol with a $m'$-round verification algorithm, for $m'=m_{max}(n_{min}(m))$. On the other hand, $m'\geq m$, and therefore $L'$ cannot be recognized by an interactive protocol with parameter $m$ either. In particular, $\cL'$ cannot be recognized by $\mathcal{P}$, a contradiction. Therefore $\cL\notin \AM[m](\sigma,\infty)$, from which it follows that $\cL\notin \AMH(\sigma,\infty)$, and therefore 
\[
\cL\notin \AMH(o(n),\infty).
\]
We complete the proof by noticing that every language $\cL$ on 0/1-labelled oriented cycles has a 1-round locally checkable proof with $O(n)$-bit certificate, where neighboring nodes exchange $O(n)$-bit messages. That is, we observe that 
\[
\cL\in\Sigma_1\LD(O(n),O(n)).
\]
The certificate given to every node consists of an $n$-bit string $s=(s_1,\dots,s_n)$ where $s_i\in\{0,1\}$ is supposed to be equal to the input of $u_i$, for every $i=1,\dots,n$. Every node checks that this is indeed the case, and that it is given the same bit-string as the ones given to its neighbors. Finally, it checks whether $s\in\cL$. If all tests are passed, it accepts, otherwise it rejects.  
\end{proof}

We complete this section by showing that, in contrast to the previous theorem, every distributed language on regular graphs with $O(1)$-bit inputs has a locally checkable proof with $O(n)$-bit certificates, and a 2-round verifier.
 
\begin{theorem}\label{thm:insideSigma1}
Every distributed language on $d$-regular graphs with $O(1)$-bit input labels  belongs to 
$
\Sigma_1\LD(\widetilde{O}(n),O(dn)),
$
with a verifier performing two rounds. 
\end{theorem}

\begin{proof}
Let $\cL$ be a distributed language on regular graphs, and let $(G,\id,x)\in \cL$. Let $d$ be the degree of $G=(V,E)$. The prover acts as follows. The set $E$ of $nd/2$ edges of $G$ is partitioned into $d$ sets $E_1,\dots,E_d$, each of cardinality at most $\lceil \frac{n}{2}\rceil$. Let every node $u$ pick a set $E_i$, $i\in\{1,\dots,d\}$ uniformly at random, and let $\mathcal{E}_u$ be the event ``the union of the sets picked by $u$'s neighbors is $E$''. We have 
\[
\Pr[\mathcal{E}_u]=d! \Big ( \frac{1}{d} \Big )^d = \frac{\sqrt{2\pi d}}{e^d} \, \Big (1+O\Big(\frac{1}{d}\big)\Big). 
\]
By repeating $c\log n$ times the experiment, we get 
\[
\Pr[\overline{\mathcal{E}_u}] \leq \Big (1-\frac{\sqrt{2\pi d}}{e^d} \, \Big (1+O\Big(\frac{1}{d}\Big) \Big) \Big )^{c\log n}.
\]
It follows that $\Pr[\overline{\mathcal{E}_u}] \leq \frac{1}{n^2}$ for $c$ large enough. Hence, $\bigwedge_{u\in V}\mathcal{E}_u$ holds with high probability. This shows that the prover has a way to assign the sets $E_i$, $i=1,\dots,d$ to the nodes such that the verifier running at every node $u$ can gather the sets of edges stored at its neighbors, and reconstruct the graph $G$ from these sets. 

By the same reasoning, it can be shown that the set of $n$ input labels can be split in $d$ sets $I_1,\dots,I_d$, and distributed evenly to the nodes by the prover, with $O(\log n)$ sets per node, so that every node can recover the entire set of input labels from its neighbors. 

Overall, with $\widetilde{O}(n)$-bit certificate, all the nodes can recover what is supposed to be the input configuration $(G,\id,x)$. Every node then performs a second round, enabling all the nodes to check that they collectively agree on $(G,\id,x)$. (If there is disagreement between two neighboring nodes about the given input configuration, these two nodes reject). A node that agree on  $(G,\id,x)$ with all its neighbors checks that this configuration agrees with its view at distance~1 in the actual network, i.e., that its incident edges and input are as specified by $(G,\id,x)$. Finally, a node that passed all these test completes the verification by checking whether $(G,\id,x)\in \cL$, and accept or reject accordingly.  

The protocol is correct as, if $(G,\id,x)\notin \cL$, then the verifier must provide all nodes with the same input configuration $(G',\id',x')\in \cL$ for expecting all nodes to accept. However, since $(G',\id',x')\neq  (G,\id,x)$, there exists at least one node $u\in V$ satisfying either $E'(u)\neq E(u)$ or $x'(u)\neq x(u)$, leading this node to reject ($E'(u)$ and $E(u)$ denotes the edges incident to $u$ in $G'$ and $G$, respectively). 
\end{proof}

Since $\Sigma_1\LD \subseteq  \AM \cap \MA$, an immediate corollary of this theorem is that every distributed language on $d$-regular graphs with $O(1)$-bit inputs belongs to $\AM(\widetilde{O}(n),O(dn))$. 
Using a known $\RPLS$ protocol~\cite{BFPS15}, we can also show that each such language belongs to $\MA(\widetilde{O}(dn),O(\log n))$.

\bibliography{dmabib}

%%%%%%%%%%%%%%%%%%%%%%%%%%%%%%%%%%%%%%%%%%%
\end{document}